\documentclass[pdflatex,sn-mathphys-ay]{sn-jnl}

\usepackage{graphicx}%
\usepackage{multirow}%
\usepackage{amsmath,amssymb,amsfonts}%
\usepackage{amsthm}%
\usepackage{mathrsfs}%
\usepackage[title]{appendix}%
\usepackage{xcolor}%
\usepackage{textcomp}%
\usepackage{manyfoot}%
\usepackage{booktabs}%
\usepackage{algorithm}
\usepackage{algpseudocode}
\usepackage{comment}
\usepackage{listings}%
\usepackage{bbm}
\usepackage{dsfont}
\usepackage{float}
\usepackage[table]{xcolor}

\usepackage{tikz}
\usetikzlibrary{arrows.meta, fit, positioning}

\newcommand{\N}{\mathcal{N}}
\newcommand{\C}{\mathcal{C}}

\newcommand{\I}{\mathcal{I}}
\newcommand{\E}{\mathcal{E}}
\newcommand{\R}{\mathbb{R}}
\newcommand{\parts}{\mathcal{P}}
\newcommand{\gray}{\cellcolor[HTML]{C0C0C0}}
\newcommand{\1}{\mathds{1}}

\newtheorem{assumption}{Assumption}

\theoremstyle{thmstyleone}%
\newtheorem{theorem}{Theorem}
\newtheorem{proposition}[theorem]{Proposition}%

\theoremstyle{thmstyletwo}%

\theoremstyle{thmstylethree}%

\begin{document}

\title[Equilibria in Routing Games With Connected Autonomous Vehicles May Not Be Strong, as the Exclusive Coalitions May Form]{Equilibria in Routing Games With Connected Autonomous Vehicles May Not Be Strong, as the Exclusive Coalitions May Form}


\author*[1,3]{\fnm{Anastasia} \sur{Psarou}}\email{anastasia.psarou@uj.edu.pl}
\equalcont{These authors contributed equally to this work.}

\author[2]{\fnm{Natello} \sur{Descormier}}\email{natellodescormier@free.fr}
\equalcont{These authors contributed equally to this work.}

\author[1]{\fnm{Rafał} \sur{Kucharski}}\email{rafal.kucharski@uj.edu.pl}

\affil*[1]{\orgdiv{Faculty of Mathematics and Informatics}, \orgname{Jagiellonian University}, \orgaddress{\street{Profesora Stanisława Łojasiewicza 6}, \city{Kraków}, \postcode{30-348}, \country{Poland}}}

\affil[2]{\orgdiv{École Polytechnique}, \orgname{Institut Polytechnique de Paris}, \orgaddress{\street{Palaiseau}, \city{Paris}, \country{France}}}

\affil[3]{\orgdiv{Doctoral School of Exact and Natural Sciences}, \orgname{Jagiellonian University}, \orgaddress{\street{Golebia 24}, \city{Kraków}, \postcode{31-007}, \country{Poland}}}


\abstract{User Equilibrium (UE) is the standard representation of the so-called routing game in which drivers adjust their route choices to arrive at their destinations as fast as possible. Asking whether UE is strong has, so far, been meaningless for human drivers, who do not form coalitions because of technical and behavioral constraints. Here, we demonstrate that this will not necessarily hold for Connected Autonomous Vehicles (CAVs), which can communicate and collaborate to form routing coalitions by exploiting adaptive signal control. Within game-theoretical formulations of Nash Equilibria, we provide a methodological framework and report the novel phenomena through a carefully designed toy-network example, in which we identify so-called \emph{clubs} emerging in different CAV penetration rates and adaptive traffic-light configurations. We detail one such club where three CAVs deviate from the UE and benefit (arrive faster). The formation of such a club has negative consequences for other users (who are not invited to join and must now travel longer) and for the system (making it suboptimal and disequilibrated), thereby triggering adaptation dynamics. These findings raise questions about equity in future mixed-autonomy traffic systems. }




\maketitle

\section{Introduction}\label{sec1}

Connected Autonomous Vehicles (CAVs) are being gradually introduced into urban environments and are expected to operate at scale alongside Human-Driven Vehicles (HDVs) in the coming years. With Waymo and Apollo Go already operating successfully in a number of locations worldwide, CAVs equipped with connectivity and optimization capabilities may form cooperative coalitions. Potential benefits of such coalitions have already been demonstrated in the contexts of platooning, lane-changing, merging, or efficiency at traffic lights \citep{smartcities6050111, 8442706, HESHAMI2024104789, FU2023104415, friedrich2016effect}. 

In this study, we demonstrate that such coalitions can also emerge 
when CAVs coordinate to select routes, i.e., paths to their destinations. 

We depart from the classical game-theoretic equilibria and translate Nash Equilibrium and strong 
Nash Equilibria into Wardrop equilibria, followed by a methodological framework to describe and measure non-strong Nash Equilibria, with an algorithm to determine whether coalitions may form. Using a minimal, realistic example, we identify CAV-exclusive clubs that may negatively impact the system consistently across settings.

We simulate a traffic network using the microscopic traffic simulator SUMO \citep{SUMO}, where both HDVs and CAVs make day-to-day routing decisions. Beyond demonstrating the existence of the phenomenon in a minimal example, we show it appears across different CAV penetration rates, traffic signal configurations, and HDVs’ route adaptation.

To better understand this unreported phenomenon, we experimentally demonstrate that: 
\begin{enumerate}
\item Coalitions of CAVs may emerge from their self-interested routing behavior in traffic networks. 
\item Their impact on system performance, CAVs outside coalitions (not invited), and HDVs (unable to join) can be negative, as their travel times will increase.
\item Coalitions can be benefit all members' payoffs (everyone may arrive faster), not just the average.  
\item Such coalitions may be unstable and induce long-term dynamics.

\end{enumerate}

\begin{figure*}[t]
    \centering
    \includegraphics[width=\linewidth]{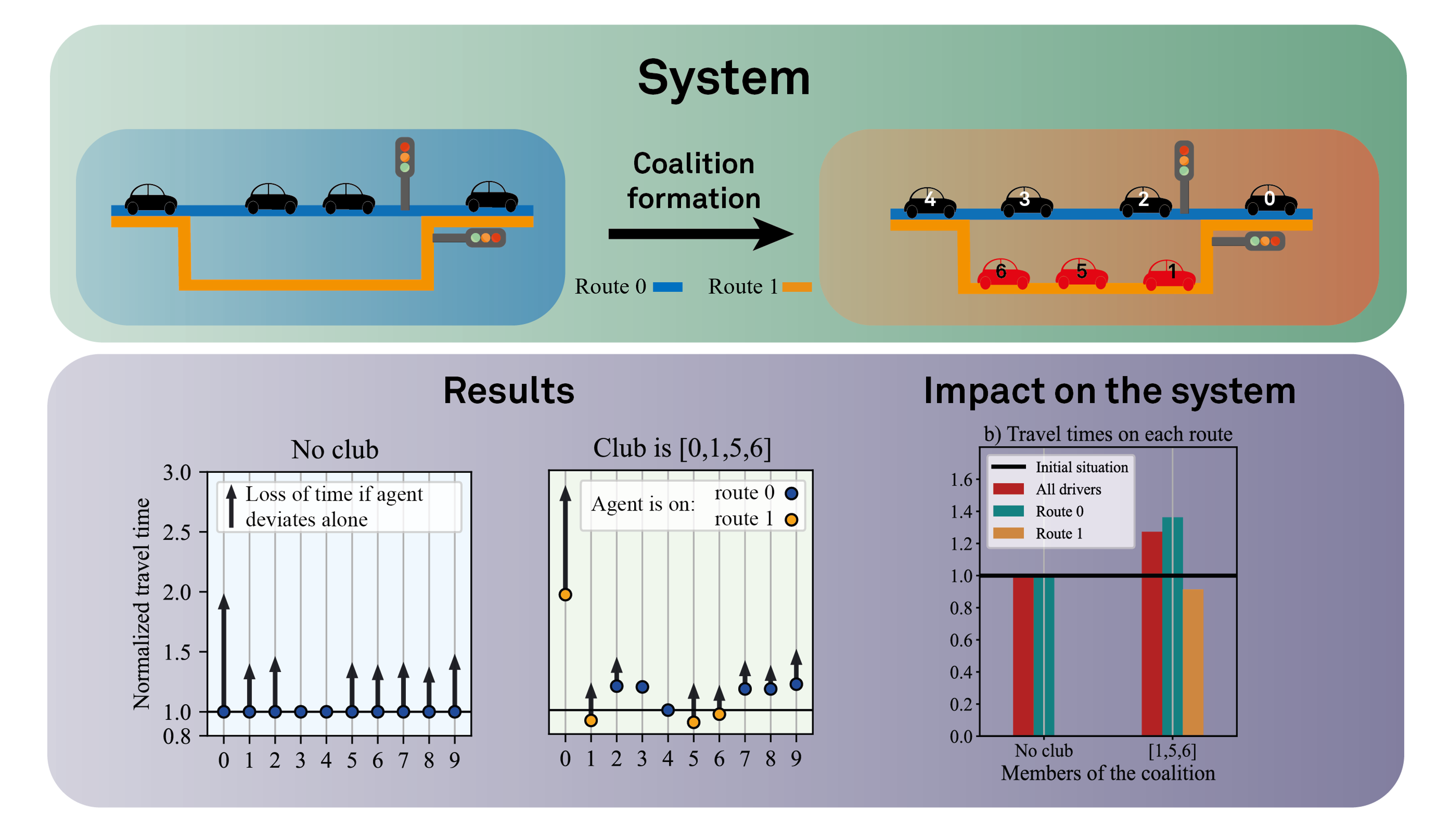}
    \caption{\textbf{CAV coalition formation.} In our demonstration, we depart from User Equilibrium (top left) on a simple Two Route Network, where each driver decides to go straight (Route 0), as traveling Route 1 would be slower for each of them. When 10 out of 15 vehicles become autonomous vehicles, they reconsider their choices. None of them would benefit by individually deviating to Route 1 (bottom left); however, vehicles 1, 5, and 6 discover that by deviating together, they can arrive faster (bottom center). Not only on average, but also individually. This benefit for the few invited club members is at the cost of other users and system performance (bottom right).}
    \label{fig:1overview}
\end{figure*}

\section{Background}
In urban traffic networks, HDVs make daily routing decisions to reach their destinations. Every day they commute, i.e., travel from home to work, by choosing among multiple alternative routes. These payoff-maximizing routing decisions are typically guided by habit, information, or exploration aimed at identifying better alternatives, given expectations about the choices of other drivers \citep{inductive_bias_bounded_rationality}. Over time, the aggregate effect of such self-interested and uncoordinated decisions is commonly assumed to converge to a Wardrop equilibrium (or UE), where no driver can unilaterally reduce their travel time by switching routes \citep{WARDROP1952, wardrop_jose_correa}. In its standard nonatomic formulation, Wardrop equilibrium assumes that individual route choices do not affect network congestion or travel times. In contrast, the present study models traffic as a finite atomic routing game, where individual route choices affect travel times. We therefore use Nash equilibrium (NE) to identify route choices in which no driver can reduce their travel time by unilaterally switching routes \citep{Nash1950}.


In the classical, well-studied setting, HDVs act independently, without the means, motivation, or infrastructure to coordinate at a meaningful scale. As a result, the notion of coalitions, groups of players coordinating their route choices to reduce travel times \citep{Peters2008}, has not had practical meaning so far. One widely considered state of such a system was a grand coalition, where all the drivers contribute to the total travel time (or mean), which they minimize to reach the System Optimum \citep{WARDROP1952}. In such (equally hypothetical) settings, coalition members would collaborate to reduce the total costs by minimizing the so-called "\textit{Price of Anarchy}" \citep{correa2008geometric, papadimitriou2001} (difference between total costs at UE and SO). However tempting for policymakers, SO remained unattainable in practice \citep{hoffmann2025wardropian}.

The emergence of \emph{CAVs} is poised to change this landscape. Unlike HDVs, CAVs can leverage real-time information, automation, and communication technologies to make fast and coordinated decisions \citep{coordination_AVs}. 

\section{Related works}

Coalition formation among CAVs has been studied primarily for local coordination tasks, such as platooning, convoying, and lane changing, where cooperation enhances individual and network performance. Platoon formation has been modeled as a coalition-formation game improving flow and safety \cite{smartcities6050111,8442706}. Similar ideas support cooperative lane-changing in mixed traffic, demonstrating higher throughput and lower congestion than to non-cooperative baselines \cite{HESHAMI2024104789,FU2023104415}. 

For route choice, prior studies typically cast CAVs as levers for system improvement. With centrally assigned CAVs steering mixed traffic toward more efficient states, sometimes with limited CAV sacrifice \cite{lazar2021learningdynamicallyrouteautonomous,GUO2022103726,10422527}. These approaches treat cooperation as benevolent or centrally guided, rather than as strategic, self-interested coalition behavior that may harm others. \cite{Kim2005RealtimeRouting} shows how information availability can change routing policies and outcomes, building on this, \cite{Spana2022InfoPerturbation} formulate coordinated routing for connected vehicles in a mixed-strategy congestion-game setting, further combined with route guidance for real-time network control \cite{Menelaou2022JointRGDM}.

In routing games, Wardrop user equilibrium (UE, NE) \cite{WARDROP1952} generally differs from the system optimum, with the gap captured by the \emph{price of anarchy} \cite{papadimitriou2001,pigou1920,roughgarden}. We use \emph{strong equilibrium} as the benchmark for profitable joint deviations: a joint action $x$ is strong if no coalition can deviate so that all its members strictly improve \cite{Aumann1959,BERNHEIM19871}.

Our analysis draws on two concepts from coalition theory. \textit{(i) Cooperative)} The core and the Shapley value characterize when a surplus-sharing coalition can exist and how to divide its gains; balancedness and convexity provide existence and stability guarantees \cite{vNM1944,Shapley1952,Shapley1967,Scarf1967,Shapley1971,AumannDreze1974,Myerson1977,Owen1977}. \textit{(ii) Non-cooperative)} Endogenous coalition-formation models explain how such coalitions arise and persist in equilibrium, and when joint deviations are credible (e.g., coalition-proof and farsighted stability) \cite{HartKurz1983,Bloch1996,Yi1997,RayVohra1999,BernheimPelegWhinston1987,Chwe1994}.

To the best of our knowledge, no prior work demonstrates the emergence of CAV coalitions that coordinate their route choices 
within a microscopic, mixed-traffic simulation where both CAVs and HDVs make self-interested routing decisions. In this proof-of-existence study, we show that CAV \emph{clubs} can deviate from UE and improve payoffs, benefiting their members while imposing longer travel times on others and harming overall system performance. We further show that such coalitions can emerge in different CAV penetration rates and traffic signal configurations and can persist even after HDVs adapt their route choices.

\section{Methodology}

We first formalize the classic routing game, establish the conditions for Nash and strong Nash equilibria, and provide an algorithm to find potential coalitions. We follow by defining the club (as we will call the CAVs' coalition) and introduce the club properties. We then discuss the dynamics induced by forming the club and use a graph representation of the state transitions to identify strong equilibria of the introduced routing games.

\subsection{Coalitional game}

Let $G = (\N,X,(g_{i})_{i \in \N})$ be a normal-form, nonzero-sum game \cite{Peters2008}. For demonstration, we consider the binary game with $\{0,1\}$ choices, as depicted in Fig.\ref{fig:1overview} for the TRY network, and leave the generalization for future work. The game $G$ is then formalized as follows:

\begin{itemize}
    \item $\N = [0,\dots, N-1]$ is the set of \textbf{players}.
    \item $X = \{0,1\}^{\N}$ is the \textbf{action space}. In a two-route game, the \textbf{joint action} of all players is a binary vector $x \in X$ of length $N$ (mind $|X| = 2^\N$). For $i \in \N, x_{i}$ is the \textbf{action} of player $i$ (we denote $\bar{x}_i = 1 - x_i$ the opposite action). This can be written as $x = (x_i,x_{-i})$, where $x_{-i}$ refers to the actions within $x$ made by other players. 
    \item $g_{i}(x)$ is the \textbf{payoff} precomputed for a given joint actions $x \in X$ and a given $i \in \N$. In this paper, we consider deterministic payoffs, which in routing are negative travel costs, simplified here to negative travel times. Thus, by maximising their payoff, each player reduces their travel time.
    \item the game $G$ has a \textbf{payoff matrix} $g = (g_{i}(x))_{\substack{i\in \N \\ x \in X}} \in \R^{\N \times X}$:

\end{itemize}

A joint action $x \in X$ is a \textit{Nash equilibrium} if no player can improve its payoff by changing its own action while other players keep the same decision. Let $NE \subset X$ be the set of Nash equilibria, formally defined as such: 

\begin{equation}
    NE = \Big\{ x \in X  : \forall i \in \N, g_i(x) \geq g_i\big((\bar{x}_i,x_{-i})\big)\Big\}
\end{equation}


Given a joint action $x \in X$, we call a \textbf{coalition} any non-empty subset $C \subset \N$ whose members decide to deviate simultaneously from $x$. The set of possible coalitions is $\C = \parts(\N)^* = \{C \subset \N : |C| > 1\}$. $x$ is a \textbf{strong equilibrium} \cite{Aumann1959, BERNHEIM19871} if there exists no coalition such that all members of the coalition improve their own payoff when the deviation happens. Let $SE \subset X$ be the set of strong equilibria:

\begin{equation*}
    SE = \Big\{ x \in X  : \forall C \in \C, \exists i \in C, g_i(x) \geq g_i\big((\bar{x}_C,x_{-C})\big)\Big\}
\end{equation*}


Once a payoff matrix has been precomputed, $NE$ and $SE$ can be found 
using \textbf{Algorithm \ref{alg:SE}.} This exhaustive enumeration 
of all $2^N$ joint actions is tractable only for small $N$; for larger 
networks, coalition-search methods that identify 
profitable deviations without full enumeration may be explored \cite{295240.295258, 
1619995.1620022}.



\begin{algorithm}[!t]
\caption{Deciding whether $x \in X$ is a Nash and/or a strong equilibrium}
\label{alg:SE}
\begin{algorithmic}[1]
\Require A payoff matrix $(g_i)_{i \in \N} \in (\R^X)^{\N}$
\State $\I \gets \varnothing$
\For{each $C \subset \N$ such that $|C| \geq 1$}
\State $y \gets (1 - x_i \text{ for } i \in C, x_i \text{ for } i \notin C)$ \Comment{club deviates}
\If{$\forall i \in C, g_i(x) < g_i(y)$}
\State $\I \gets \I \cup \{C\}$
\EndIf
\EndFor
\State \textbf{return} $\I$ \Comment{$x \in SE \iff \I = \varnothing$}
\State \Comment{$x \notin NE \iff \exists C \in \I, |C| = 1$}
\end{algorithmic}
\end{algorithm}



\subsection{Club - coalition of deviating CAVs}
\label{sec:club_definition}

We assume $G$ has at least a Nash equilibrium $x^0$, and without loss of generality, one may rearrange $X$ such that $x^0 := (0,0,\dots,0)$. We then call a \textbf{club} a group of at least two players that decide, with the knowledge of their own payoffs, to ally together and change their action simultaneously, i.e., deviate from $NE$. Let $I \in \C$, then the club $I$ may form only if:

\begin{equation}
\label{eq:clubexistence}
\left\{
\begin{array}{l}
g_i((1_i,0_{-i})) \leq g_i(x^0),\\
g_i((1_I,0_{-I})) > g_i(x^0),
\end{array}
\right.
\ \forall i \in I
\end{equation}


Given a payoff matrix $g$, we define $\I$ as the set of such clubs. From this definition, it follows that $\I$ has the following properties:
\begin{itemize}
    \item $\I = \varnothing \iff x^0 = (0,0,\dots,0) \in SE$
    \item $\I \neq \varnothing \implies \forall I \in \I, |I| \geq 2$.
\end{itemize}

To ease the notation, for our two-route case we abbreviate to $\1_I$ the joint action $x \in X$ such that all members of $I$ have made the choice 1 (deviated) and other players are still on choice 0 (follow NE):

\begin{equation}
 \1_I = x_i \quad \text{such that} \quad x_i = 
\begin{cases}
1, & \forall\, i \in I, \\
0, & \forall\, i \notin I
\end{cases}
\end{equation}

\subsection{Club properties}

Once a club forms and its members all deviate to Route 1, the new state no longer has to remain a Nash equilibrium. Players outside the club may now choose to deviate as well (regardless of not being invited during the club's formation), hoping to improve their payoffs, potentially at the cost of the club's members.

To better understand it, we define two stability conditions for a club, which, combined, are equivalent to Nash equilibrium.
\begin{itemize}
    \item The first condition is \textbf{"internal stability"} or "individual rationality" \cite{Shenoy1979}, under which no member of the club has an incentive to leave, i-e,
    \begin{equation}
        \forall i \in I, g_i(\1_I) \geq g_i(\1_{I\backslash \{i\}})
        \label{eq:internalstab}
    \end{equation}
    \item However, once other players have noticed that a club has formed, the club has no power to stop them from deviating as well. To prevent anyone else from being willing to join, we introduce an \textbf{"external stability"} condition:
    \begin{equation}
        \forall j \notin I, g_j(\1_I) \geq g_j(\1_{I\cup \{j\}})
        \label{eq:externalstab}
    \end{equation}
    \end{itemize}

Although not explored here, one could also investigate whether a subset $S$
of the players in $I$ wishes to remove the others, $I \setminus S$, from the
club \cite{doi:10.1142/S0219198909002352}, which requires searching over
subsets, which is a costly operation. Similarly, we do not consider that club members
may want to invite a group of non-members \cite{3535850.3536018}. Moreover,
we do not re-apply Algorithm~1 to search for new coalitions once a stable
club has formed; hence, we do not claim that the final state is a strong
equilibrium in the formal sense \cite{Aumann1959}.

\begin{figure}[h]
\centering
\resizebox{!}{10cm}{%
\begin{tikzpicture}

\path[use as bounding box] (-6.5,-23) rectangle (11.6,1.5);

\tikzset{
    bigbox/.style={draw, fill=gray!20, font=\Huge, minimum size=16mm},
    coalitionbox/.style={draw, fill=gray!20, font=\Huge, minimum height=14mm, minimum width=38mm},
    smallcoalitionbox/.style={draw, fill=gray!20, font=\Huge, minimum height=14mm, minimum width=38mm},
    arrowlabel/.style={font=\Huge},
    smallarrowlabel/.style={font=\Huge},
    leaderbox/.style={
    bigbox,
    very thick,
    double,
    double distance=1.2pt
    },
    leadergroup/.style={
        draw,
        densely dashed,
        rounded corners,
        inner sep=3mm
    },
    leaderlabel/.style={
        font=\Large,
        align=center
    },
    secondordernode/.style={
    draw,
    dashed,
    thick,
    fill=gray!10,
    font=\Huge
},
secondorderregion/.style={
    draw,
    densely dashed,
    rounded corners,
    thick,
    inner sep=5mm
},
secondorderlabel/.style={
    font=\Huge,
    align=center
},
secondorderarrow/.style={
    -{Stealth[length=10pt]},
    dashed,
    thick,
    shorten >=2pt,
    shorten <=2pt
}
}

\node[leaderbox] (I1) at (-3.8,0.02) {$1$};

\node[leaderbox] (I5) at (3.8,0.02) {$5$};

\node[
    leadergroup,
    fit=(I1) (I5),
    label={[leaderlabel]above:{candidate leaders / initiators}}
] (LeaderGroup) {};

\node[coalitionbox] (I15) at (0,-5.5) {$I=\{1,5\}$};

\draw[-{Stealth[length=10pt]}, shorten >=2pt, shorten <=2pt]
    (I1) to node[midway, above, xshift=-22mm, yshift=-12, arrowlabel] {$5$ is invited} (I15);

\draw[-{Stealth[length=10pt]}, shorten >=2pt, shorten <=2pt]
    (I5) to node[midway, above, xshift=22mm, yshift=-12, arrowlabel] {$1$ is invited} (I15);

\node[draw=red, thick, font=\Huge,           minimum height=14mm,
    minimum width=38mm] (I1504) at (-3.8,-10.2) {$I=\{1,5,0,4\}$};

\node[red, font=\Large, anchor=east] 
    at ([xshift=-3mm]I1504.west) {$E(\mathcal{G}^{I})$};
    
\draw[-{Stealth[length=10pt]}, shorten >=2pt, shorten <=2pt]
    (I15) to node[midway, above, xshift=-24mm, arrowlabel] {$0,4$ are invited} (I1504);

\begin{scope}[shift={(3.8,-10.2)}]

\node[
    draw,
    fill=gray!20,
    font=\Huge,
    minimum height=14mm,
    minimum width=38mm
] (I156) at (-0.8,0) {$I=\{1,5,6\}$};

\node[
    draw=red,
    thick,
    font=\Huge,
    minimum height=14mm,
    minimum width=38mm
] (A) at (-0.8,-2.7) {$\{0,1,5,6\}$};

\node[
    draw,
    font=\Huge,
    minimum height=14mm,
    minimum width=38mm
] (B) at (6,0) {$\{1,5,6,7\}$};

\node[
    draw,
    font=\Huge,
    minimum height=14mm,
    minimum width=38mm
] (C) at (6,-2.7) {$\{0,1,5,6,7\}$};

\node[secondordernode] (E) at (-0.8,-9.8) {$\{0,1,5,6\}, 13$};

\draw[-{Stealth[length=10pt]}]
    (I156) to node[left=0.7mm, font=\Large] {$0 \in I^+$} (A);

\draw[-{Stealth[length=10pt]}]
    (I156) to node[above=0.7mm, font=\Large] {$7 \in I^+$} (B);

\draw[-{Stealth[length=10pt]}]
    (B) to node[left=0mm, font=\Large] {$0 \in \{1,5,6,7\}^+$} (C);

\draw[-{Stealth[length=10pt, open]}, dashed]
    (C) to node[below=0.7mm, font=\Large] {$7$ quits} (A);

\draw[secondorderarrow] (A) -- (E);

\node[draw=red, thick, dotted, fit=(A) (C), xshift=0mm,
    yshift=-1.5mm, inner sep=4mm] (X) {};

\node[red, font=\Large, anchor=east] at ([xshift=-0.5mm]X.west) {$E(\mathcal{G}^{I})$};
\end{scope}

\node[secondordernode] (I1504copy) at (I1504 |- E)
    {$\{1,5,0,4\}, 13$};

\draw[secondorderarrow] (I1504) -- (I1504copy);

\node[
    blue,
    font=\Large,
    align=center,
    below=6mm of I1504copy
] {stable after\\HDV deviation\\(noone wants to deviate)};

\node[
    blue,
    font=\Large,
    align=center,
    below=6mm of E
] {stable after\\HDV deviation\\(noone wants to deviate)};

\coordinate (HumanLabelY) at (0,-18);

\node[
    font=\Large,
    align=center,
    anchor=west
] at ([xshift=1mm]I1504 |- HumanLabelY)
    {HDV $13$\\deviates alone};

\node[
    font=\Large,
    align=center,
    anchor=west
] at ([xshift=1mm]A |- HumanLabelY)
    {HDV $13$\\deviates alone};

\node[
    secondorderregion,
    fit=(I1504copy) (E)
] (SecondOrderRegion) {};

\node[
    secondorderlabel,
    anchor=south
] at ([yshift=21mm]SecondOrderRegion.north)
{second-order \\ phenomena};

\draw[-{Stealth[length=10pt]}, shorten >=2pt, shorten <=2pt]
    (I15) to node[midway, above, xshift=24mm, arrowlabel] {$6$ is invited} (I156);
    
\end{tikzpicture}%
}

\caption{Illustration of the club dynamics graph $\mathcal{G}^I$ with $I = \{1,5\}$, from which two coalitions are identified using \textbf{Algorithm 1}. Either player 1 or 5 can act as the leader and identify that forming the club $I=\{1,5\}$ would reduce their travel times. In the first branch, players 0 and 4 are consequently invited to deviate, resulting in the terminal coalition {$\{0,1,5,4\}$}, which corresponds to a Strong Nash equilibrium. In the second branch, player 6 is invited to join, and the club $I=\{1,5,6 \}$ is formed. Then, both players 0 and 7 would like to join it, given that it is formed (eq. \ref{eq:externalstab}). However, when both of them join, it is no longer attractive to player number 7. This leads to a single terminal node, the coalition {$\{0,1,5,6\}$}. Subsequently, in both formed coalitions, HDV 13 may adapt and improve their travel time by individually deviating. After joining, both coalitions remain stable.}


\label{fig:2graph} 
\end{figure}

\subsection{Club dynamics graph}

To represent the dynamics triggered by the club formation, we propose the following graph representation.
For any club $I$ identified with \textbf{Algorithm 1}, we build a directed graph $\mathcal{G}$ where nodes are possible coalitions $C$ extending $I$, not identified with the algorithm. They are obtained from players who would like to join the club $I^+ = \left\{j \notin I : g_j(\1_I) < g_j(\1_{I\cup \{j\}}) \right\}$ (from eq. \ref{eq:internalstab}). The nodes $V(I) = \left\{I \cup \{j\} | j \in I^+ \right\}$. Then, for every $I \in \I$, we build a directed graph $\mathcal{G}^I=(V^I,E^I)$, rooted in $I$ such that \[V^I \subset \C, I \in V^I, E^I = \left\{ (C,C') : C \in V^I, C' \in V(C)\right\}\]
Therefore, for every $J \subset \N$, $J \in V^I$ iff there exists a sequence of coalitions $J_0,J_1,\dots,J_n$ such that $J_0 = I$, $J_n = J$ and $\forall p \in [0\dots n-1], J_{p+1} \in V(J_p)$. Mind, that some possible coalitions $J$ more than one sequence from $I$ may exist (see Fig.\ref{fig:2graph}), making this structure a bush rather than a tree.

External vertices (leaves) of the tree $G^I$ cannot be expanded by adding a new players (i.e., they are externally stable). Further, individual players may join by adaptation making the final nodes in this representation.
Let $\E(G^I)$ be the set of external vertices of $G^I$. $\forall J \in \C$,

\begin{equation}
\begin{aligned}
    J \in \E(G^I) \implies {}&
    (1_J,0_{-J}) \in NE \\
    &{}\text{or } J \text{ is not internally stable}.
\end{aligned}
\end{equation}

The graph constructed for the club from our example is illustrated in Fig.\ref{fig:2graph}. Here, we start with two possible leaders and two branches, terminated with HDV drivers joining the coalitions.

Coalitions $I$ at the leaves of a certain $\mathcal{G}^I$, which we further call \textbf{terminal coalitions}, have two important properties:
\begin{itemize}
    \item They will not be modified by players outside the coalition. Therefore, as long as all members can find an arrangement or a way to distribute their payoff that satisfies all of them, the coalition will remain unchanged.
    \item After the members of the root coalition $I$ choose to deviate from $x^0$, terminal coalitions are attainable by a series of individual, selfish decisions by other players. Therefore, if $I$ is formed and as soon as other players are aware of this decision, the transition to one of these coalitions will proceed.
    
\end{itemize}

Based on this, we terminal leaves as a candidates to search for Strong Nash Equilibria. However, in this terminal state (of some coalition identified with \textbf{Algorithm 1} and extended along the graph $\mathcal{G}$), the new coalition opportunities may arise, and some other players may want to form a new club. In this paper, we omit those second-order phenomena and leave them for further studies. We just assert that the Strong Equilibrium may be identified at the external vertex (leaf) of the graph $\mathcal{G}$ for which applying \textbf{Algorithm 1} yields $\I = \emptyset$. Which shall be formalized more rigorously in follow-up studies.

\subsection{Formation process}
The above raises the question: If such clubs may theoretically exist, how can they be formed in practice? This we do not know; we outline one plausible realization and leave empirical validation for future work. The consider mechanism would proceed as follows:

\begin{enumerate}
\item First, let's identify some CAVs as a \emph{leader} $L$ of the future coalition. 
\item We assume that the leader is aware of the full payoff matrix $g$, not only in the UE ($x^0$) but also for any other hypothetical joint action $x \in X$ , retrieved from simulations. 
\item Leader can apply \textbf{Algorithm \ref{alg:SE}} to identify coalitions $\I$.
\item Leader can decide which coalition $i \in \I$ it wants to form (obviously among those it belongs to: $L \in I$) and submits invitations to form a club to other CAVs. 
\item $L$ will suggest the club $I \in \I$ if $L \in I$, i.e. if  $g_L(\1_I) > g_L(x^0)$, 
\item All members of $I$ shall accept the invitation, since $\forall j\in I, g_j(\1_I) > g_j(x^0)$.
\item For now, we assume the club formation takes place overnight (say Monday night) and is executed on the next day (say Tuesday morning commute). On this day, members of the coalition deviate, i.e., choose Route 1, unlike the day before (on Monday, everyone traveled on Route 0). In practice, signal adaptation could also occur within the same day if earlier-departing members deviate to benefit later ones.
\item The dynamic supply component (adaptive traffic light) adjusts to the new situation with some delay, so the values computed in matrix $g$ are executed with some delay, say on Wednesday (detailed in the \emph{Demonstration} section below).
\item Players not invited to join the club (both HDVs and CAVs) are first clueless and stick to choice 0. After some time they may react (adjust their choices or explore) and deviate.
\end{enumerate}

During formation, leader $L$ may consider not only clubs identified with \textbf{Algorithm \ref{alg:SE}}, but also the clubs among the vertices of $G^I$. It may choose to invite extra players to the club, one after the other, along the vertices of the graph constructed using the method described above. 
Here, we consider the leader to be interested in creating stable coalitions, rather than those that maximize its payoffs. 
Thus moving directly towards the terminal leaves of identified graph as on Fig.\ref{fig:2graph}.

While we do not claim the above process reflects how CAV coalitions will form 
in practice, we present it as one sufficiently realistic candidate 
for the existence proof. One concrete realization is a fleet operator 
(such as Waymo or Apollo Go) running overnight simulations in a 
virtual environment to reproduce the payoff matrix and identify clubs before the next commute. This requires knowledge of 
departure times, origins, destinations, and signal control logic. Which is not beyond capabilities of 
a large-scale operator with access to historical data. 
Alternatively, without a central leader, vehicles could negotiate 
directly via V2X communication, decomposing costly payoff matrix calculations to multiple players. 
The full payoff matrix assumption is challenging in larger networks when
exhaustive enumeration becomes infeasible and approximate 
coalition-search methods would be needed \cite{295240.295258, 
1619995.1620022}. Also, we assume simple deterministic joining rule, leaving strategic heterogeneous responses of players (like \cite{BrandtBullingerWilczynski2023,SaadHanBasarDebbahHjorungnes2011}) for future works. How this will play out in practice will be determined by fleet operators, 
regulators, and the market.

\section{Demonstration}

We model an abstract, yet realistic, two-route congested network, where vehicles depart from the same origin $A$ to the same destination $B$ as depicted in Fig.\ref{fig:3sumo_screenshot}. Vehicles are homogeneous in their properties (desired speed, acceleration, etc.), have predefined and fixed departure times, and the same payoff formulation (negative travel time). 

Every day of this repeated game, they face the choice between Route $0$ and Route $1$. Route $0$ is shorter than Route $1$, which generally means that, even considering the induced traffic congestion, payoffs for all the players are greater on Route $0$ (travel times are shorter). 
The traffic flow every day (every turn of this game) is reproduced using the microscopic traffic simulation tool, SUMO \citep{SUMO}, and the route choices of players are reproduced within the RouteRL framework \citep{RouteRL2024}. 

Here, the solution is intentionally trivial, and players find the Nash Equilibrium (NE), which is also System Optimal, as $\mathbf{0}^N$ where all the players choose Route 0. 
Wardrop's User Equilibrium principles \citep{WARDROP1952} in our context translate to the following properties:

\begin{assumption}
     $x^0 = (0,0,\dots,0) \in NE$
 \end{assumption}

 \begin{assumption}
     $\forall x \in X, \displaystyle\sum_{i \in \N} g_i(x^0) \geq \displaystyle\sum_{i \in \N} g_i(x)$
 \end{assumption}

 Thus, we depart from a system which is both optimal (Assumption 2) and in equilibrium (Assumption 1), which, as we demonstrate, is not strong, $x^0 = (0,0,\dots,0) \notin SE$.

In such a system, every single player has the possibility to travel either on Route 0 or Route 1 to reach the destination. In Equilibrium everyone uses Route $0$ and if the coalition emerges, they jointly \emph{deviate} from NE ($x^0$) to Route 1. Given an instance of a two-route network, there are $|X| = 2^N = 1024$ possible joint actions.

\begin{figure}[ht]
    \centering
    \includegraphics[width=0.7\linewidth]{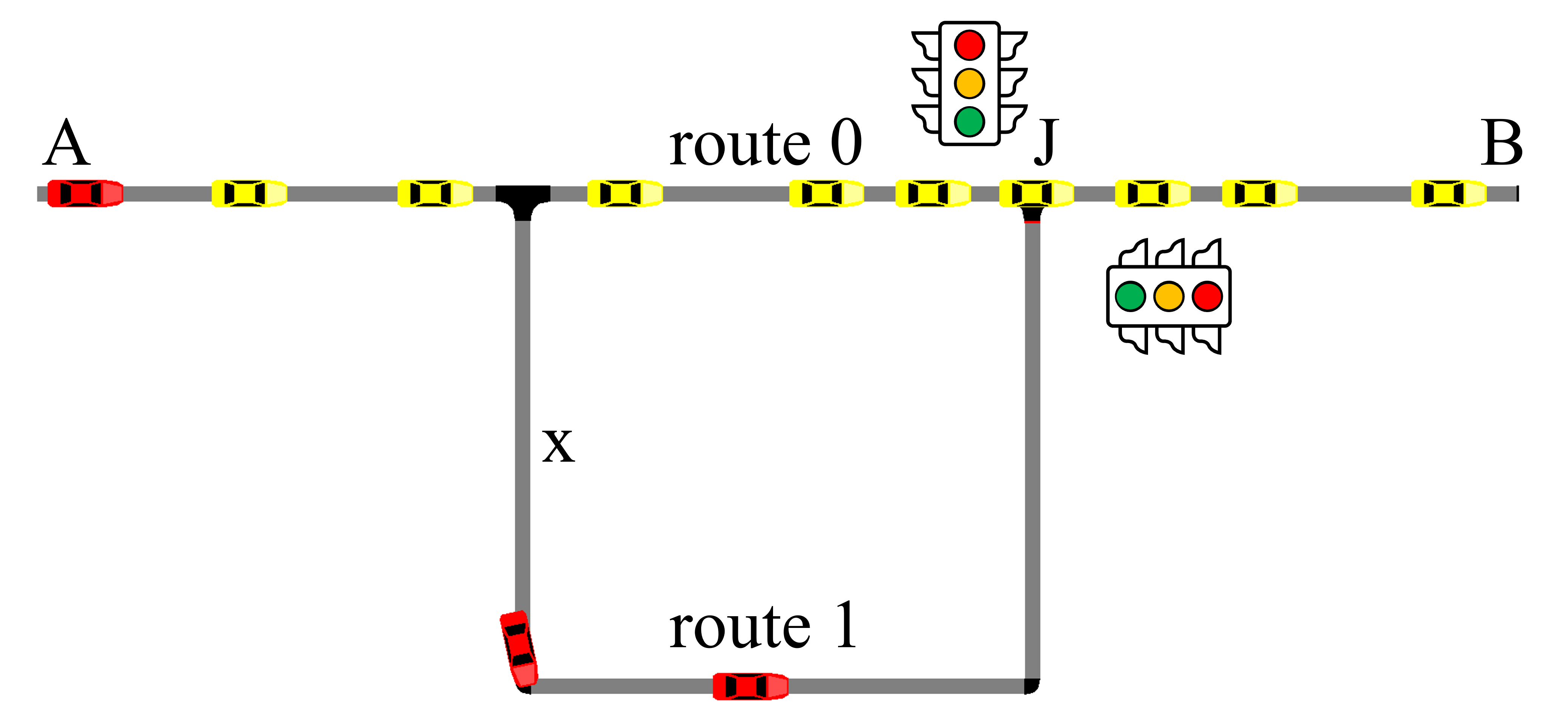}
    \caption{SUMO snapshot of the 
    traffic network employed in this work. The deviating club traveling on Route 1 is shown in red.}
    \label{fig:3sumo_screenshot}
\end{figure}

\subsection{Dynamic supply is necessary for coalition formation}

First, we sketch a proof supporting the hypothesis that the dynamic supply component is needed to form clubs.
In a similar model where the traffic light is replaced by a yield sign giving priority to Route 1 \citep{psarou2025autonomousvehiclesusingmultiagent}, the assignment of all vehicles to Route 0 is strong equilibrium, and coalitions are not identified.
We hypothesize that this can be due to the FIFO rule, or non-monotonic supply model (where increasing the route flow by at least two vehicles may decrease its travel time). Classically, $\partial t_r / \partial q_r \geq 0$: adding vehicles to a route slows it down, so coalition members deviating together to Route~1 only congest it for one another.

However, this is not the case with adaptive traffic ligths. In our case, if at least $q_1^{\mathrm{thr}}$ vehicles deviate to Route~1 the signal adapts and travel time on Route~1 drops, violating the monotonicity and making (as we demonstrate below) the coalition formation possible.

The following proposition suggests that no clubs, as defined in (\ref{eq:clubexistence}), can be found with a static supply system.

\begin{proposition}
With a static traffic light system, $x^0 \in SE$.
\end{proposition}
\begin{proof}
Let us assume $x^0 \in NE \backslash SE$. $x^0 \notin SE$ so there exists a coalition $C \in \C$ such that $|C| \geq 2$ and $\forall i \in C, g_i(x^0) < g_i(\1_C)$.
Let $i \in C$ be the last member of $C$ to drive through the network. As $x^0 \in NE$, $g_i(x^0) \geq g_i(\1_{\{i\}})$. Therefore,
\[g_i(\1_C) >g_i(x^0) \geq g_i(\1_{\{k\}})\]
However, since the Route 1 was empty in $x^0$ there would be no other vehicles in front of $i$ if it deviated alone, so player $i$ could only go faster, thus $g_i(\1_{\{i\}}) \geq g_i(\1_C)$. Arises a contradiction, therefore $x^0 \in SE$.
\end{proof}

\subsection{Dynamic Supply}
\label{sec:dynamic_supply}
Due to the aforementioned properties, we included in our setting one of the classic ways of adapting supply to demand. That is the traffic control, namely, adaptive traffic lights. Most traffic lights at signalized junctions in our cities are now adaptive meaning the green phase duration $G_r$ adjusts in response to the inflow of vehicles on route $r$.

Without exploring the variety of control methods (max pressure \citep{Varaiya2013MaxPressure}, P0 \citep{Smith1979P0Example}, multi-band \citep{Gartner1991Multiband}, etc.), we only describe the following abstract case. We introduce a rule-based signal plan with a fixed cycle of 50s. There are two intergreen phases, both fixed to 5s, and the remaining 40s are distributed between the Route 0 and Route 1 inlets. We introduce $q_1^{\mathrm{thr}}$ as the dynamic adaptation threshold variable, which specifies the minimum number of vehicles traveling on Route 1 that trigger adaptation of the traffic light duration from $G_1$ to $G'_1$. In the basic setting, the green time for Route 1, $G_1$, increases from 19 to 31s when the demand on Route 1 is at least $q_1^{\mathrm{thr}}$ vehicles (the green phase for the Route 0 inlet, $G_0$, drops from 21 to 9s respectively). Such adaptation may result in non-monotonicity of travel times, when the flow increasing by $q_1^{\mathrm{thr}}$ may reduce travel time, which (as we speculate) is a prerequisite for club formation.
For each possible joint action $x$ (note that $|X|=2^{10} = 1024$), we first determined the resulting signal setting for given $q_1^{\mathrm{thr}}$, implemented it in SUMO, and simulated to obtain the respective row of the payoff matrix $g_x$. 
Notably, the adaptation occurs overnight, i.e., the phases for Tuesday are based on the flows from Monday (in practice this adaptation can be much faster). Nonetheless, such an abstract setting is sufficient to demonstrate the formation of a club; more nuanced traffic control paradigms are left for follow-up studies.

\section{Results}

We examine whether coalitions can emerge under different combinations of green-phase durations of the adaptive traffic lights and CAV penetration rates in Fig. \ref{fig:initial_state_equilibria}. 
A total of 244 parameter combinations are identified in Fig. \ref{fig:initial_state_equilibria} for which the initial state constitutes an NE but not an SNE, allowing profitable coalitional deviations and subsequent club formation. These cases occur under different green-phase durations of adaptive traffic signals, for both CAV penetration rates considered, namely 15 players (10 CAVs and 5 HDVs) and 23 players (10 CAVs and 13 HDVs), and for dynamic traffic light threshold values $q_1^{\mathrm{thr}}$ of 2 and 3. These findings suggest that coalition formation is not restricted to a specific CAV penetration rate or adaptive traffic signal configuration, but may arise across a range of traffic conditions.

We select one of the formed coalitions and report its members' individual travel times and average CAV travel times in UE and across four different clubs resulting from CAVs joining or leaving the initial club in Table \ref{tab:1payoffs}, illustrated with relative payoff changes in Fig.\ref{fig:5deviations}. For this coalition, we additionally report average system travel times on each Route, as well as the average system travel times for CAVs and HDVs in Fig.\ref{fig:6average_travel_times}, and illustrate the travel times and the corresponding equilibrium state of each possible joint action in Fig.\ref{fig:7scatter}. Lastly, we inspect whether and how HDV choice adaptation can alter the coalition formed in Fig. \ref{fig:8human_tt}.

\begin{table}[h]
    \centering
    \caption{Payoffs for selected joint actions. Shaded entries correspond to club members. 
    The column $\bar g$ reports the average payoff over all 10 CAVs.}
    \label{tab:1payoffs}
    \begin{tabular}{c|rrrrr|rc}
        \hline
        $x$ 
        & $g_0$ & $g_1$ & $g_5$ & $g_6$ & $g_7$ 
        & $\bar g$ & State \\ 
        \hline \hline
        $x^0$ 
        & -27 & -58 & -59 & -59 & -51 
        & -48.3 & $NE$ \\

        $\1_{\{1,5,6\}}$ 
        & -77 & \gray -52 & \gray -52 & \gray -57 & -59 
        & -56.6 & \\

        $\1_{\{1,5,6,7\}}$ 
        & -77 & \gray -52 & \gray -53 & \gray -57 & \gray -52 
        & -55.8 & \\

        $\1_{\{0,1,5,6,7\}}$ 
        & \gray -53 & \gray -53 & \gray -53 & \gray -57 & \gray -71 
        & -55.1 & \\

        $\1_{\{0,1,5,6\}}$ 
        & \gray -53 & \gray -53 & \gray -53 & \gray -57 & -60 
        & -54.4 &  \\ 
        $\1_{\{0,1,5,6\}, 13}$ 
        & \gray -54 & \gray -57 & \gray -52 & \gray -56 & -60 
        & -59.9 & stable \\ 
        \hline
    \end{tabular}%
\end{table}


\begin{figure}[h]
    \centering
    \includegraphics[width=0.9\linewidth]{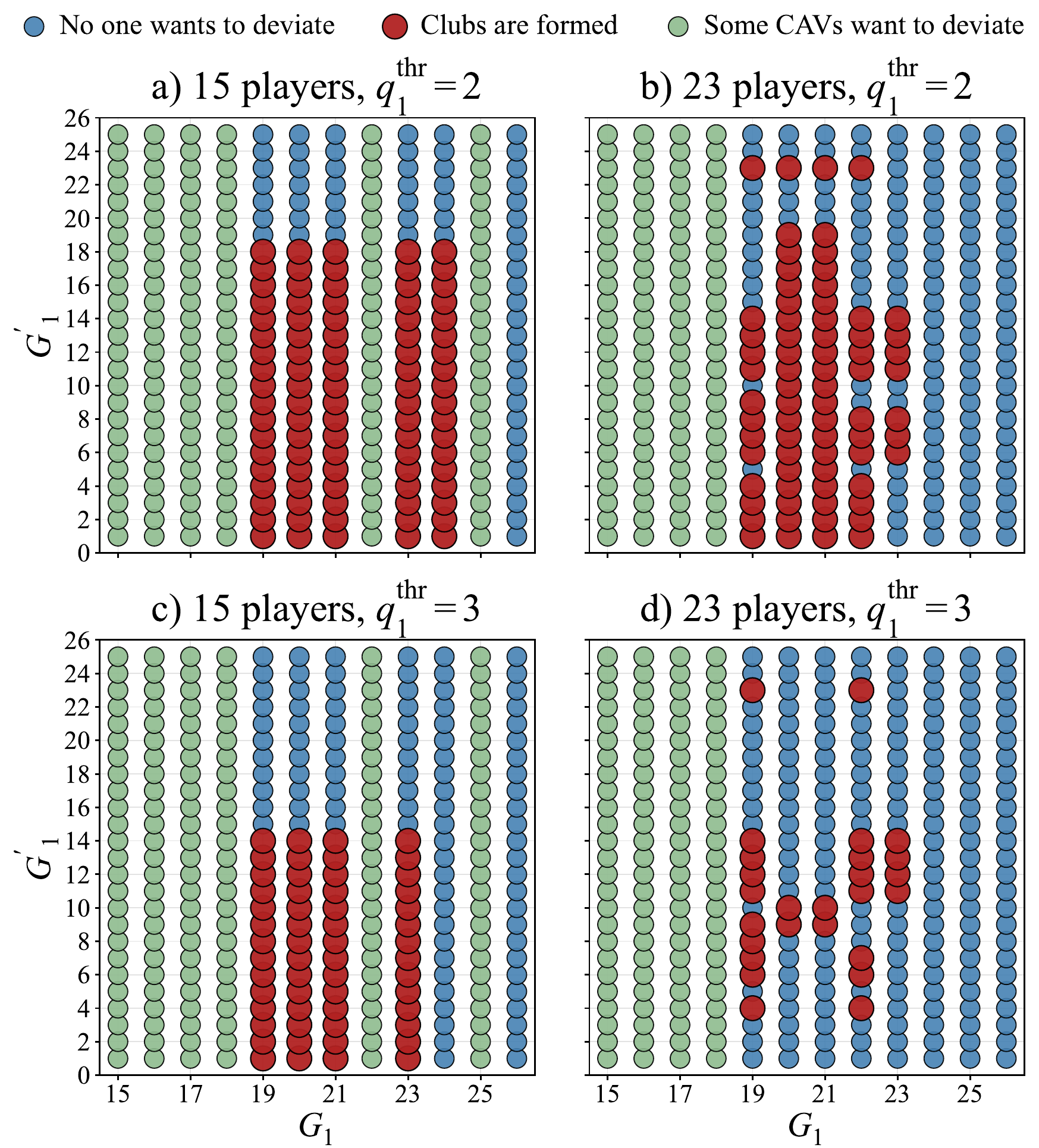}
   \caption{Clubs are identified (red markers) in a total of 244 cases under different combinations of adaptive green-phase durations of adaptive traffic lights, two CAV penetration rates of 66.6\% (15 players, 10 CAVs, 5 HDVs) and 43.5\% (23 players, 10 CAVs, 13 HDVs), and two adaptive threshold values $q_1^{\mathrm{thr}}$= 2 and 3. This suggests that club formation is not an isolated occurrence but may arise under a range of traffic conditions and adaptive traffic light settings. The x-axis represents the green-light duration when fewer than $q_1^{\mathrm{thr}}$ AVs choose Route 1, and the y-axis represents the green-light duration when $q_1^{\mathrm{thr}}$ or more AVs choose Route 1.}
   \label{fig:initial_state_equilibria}
\end{figure}

\subsection{Proof of existence}


We select one illustrative parameter combination that allows club formation, identified in Fig. \ref{fig:initial_state_equilibria} with a CAV penetration rate of 66.7\%, a $q_1^{thr}=2$, and green-phase durations $G_1, G'_1$ as described in Section \ref{sec:dynamic_supply}. Then, we apply  \textbf{Algorithm 1} and identify one coalition $\I = \{\{1, 5\}\}$, which could be initiated by either player 1 or 5 being a Leader of such a formation process. Then, the coalition can evolve in two ways: by inviting either players 0 and 4, or player 6 to join.
The graph $\mathcal{G}^I$ representing these dynamics is depicted in Fig. \ref{fig:2graph}. 
The coalition $\{1,5,0,4\}$ is strong among CAVs, i.e., no other CAVs want to join. If we include HDV adaptation in the process, player 13 will benefit from joining, which terminates the considered dynamics.

We will focus, however, on the case where player 6 joins the initial club and examine the resulting coalition-formation dynamics in greater detail. The club $I =\{1,5,6\}$ consists of 3 CAVs who would all reduce their travel time if they deviate from the NE (switch to Route 1). 
Specifically, 
CAV 1 travels 52 instead of 58 seconds, CAV 5 travels 52 instead of 59 seconds, and CAV 6 travels 57 instead of 59 seconds. The mean travel time increases from 48 to 57s, and CAV 0 travels 77 instead of 27 seconds now (as reported in Table \ref{tab:1payoffs}).

\begin{figure}[h]
    \centering
    \includegraphics[width=1\linewidth]{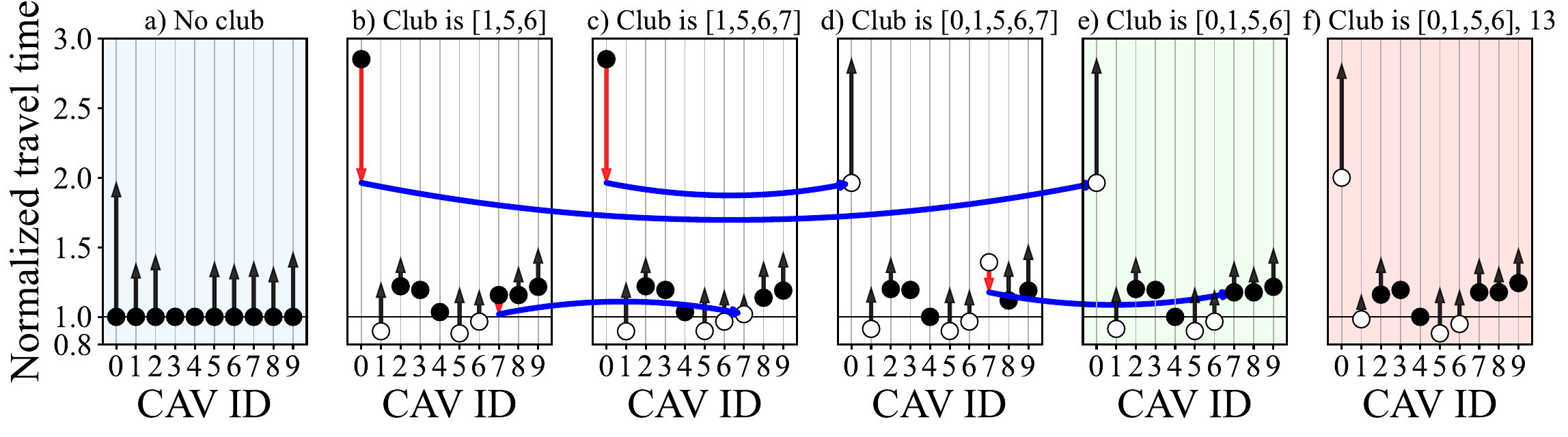}
    \caption{Travel times of 10 CAVs in NE $x^0$ (a) and in four considered club compositions $\1_{\{1,5,6\}}, \1_{\{1,5,6,7\}}, \1_{\{0,1,5,6,7\}}, \1_{\{0,1,5,6\}}, \1_{\{0,1,5,6\}, 13}$ (b-f).
    Travel times are normalized across the 5 plots, such that the travel time of all players in $x^0$ is 1.
    Black (respectively red) arrows show how much time each player would lose (respectively gain) by changing route unilaterally from a given joint action. The absence of red arrows is therefore equivalent to the joint action being a Nash equilibrium (a), (e), and (f).   
    Club members (deviating to Route 1) are marked with white circles, and those not invited (remaining at Route 0) with black ones. 
    Blue arrows connect cases where a player can gain time by changing route unilaterally to the state resulting from this change (vertices of graph $\mathcal{G}$).
    }
    \label{fig:5deviations}
\end{figure}

\begin{figure}[h]
    \centering
    \includegraphics[width=\linewidth]{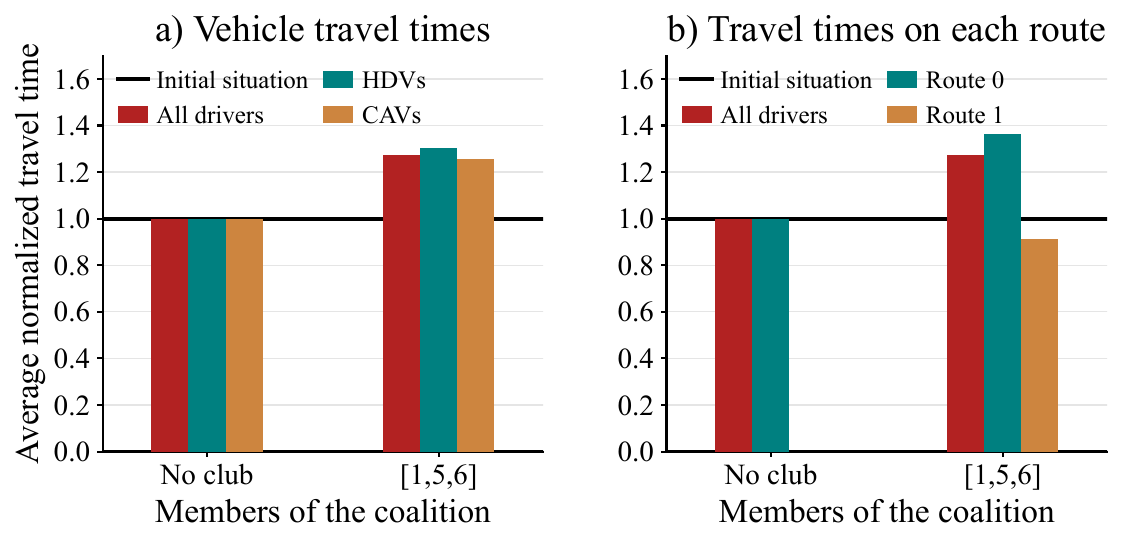}
    \caption{\textbf{System-wide average travel times.}  a) Travel times (payoffs) experienced by various groups of players (HDVs, CAVs, and all) before (left) and after coalition $\{1,5,6\}$ is formed (right). b) Average travel times experienced on respective routes. Values normalized to travel times in UE.}
    \label{fig:6average_travel_times}
\end{figure}

\begin{figure}[h]
    \centering
    \includegraphics[width=1\linewidth]{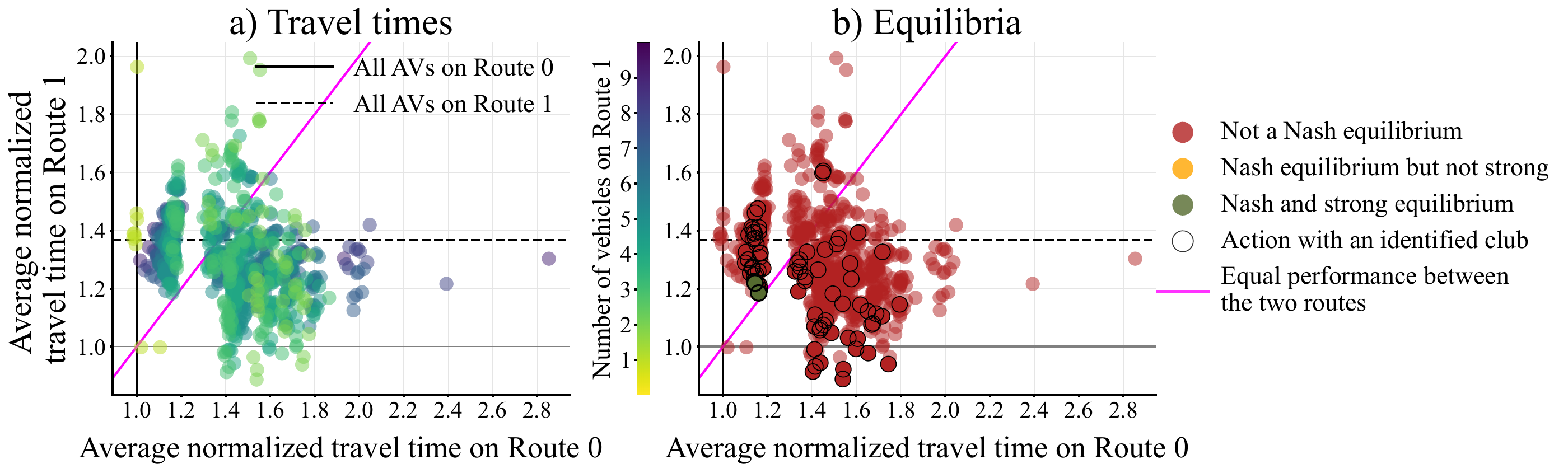}
    \caption{\textbf{System states across various joint actions.} Each dot denotes a joint action and its coordinates are the average travel times on Route 0 (x-axis) and Route 1 (y-axis), normalized such that the travel time of all vehicles in $x^0$ is 1. The pink diagonal represents equal travel times on both routes. On panel a) dots are colored proportionally to the split of vehicles between routes. On panel b), colors denote various equilibrium states: red when it is not Nash, yellow for Nash, and green for strong Nash. We identified several joint actions where we can identify a club (marked with black circles on b); however, some of the joint actions were not in Nash Equilibrium (red), including clubs $\{1, 5\}, \{1, 5, 6\}, \{1, 5, 6, 7\}, \{0, 1, 5, 6, 7\}$ and only two actions were Strong Nash (green), corresponding to clubs $\{0, 1, 5, 6\}, \{0, 1, 5, 4\}$}.
    \label{fig:7scatter}
\end{figure}

These changes are shown in Fig. \ref{fig:5deviations}b: CAVs 1, 5, and 6 (marked with white circles) travel now faster than in $x^0$ NE (values below 1), but almost all other players face longer times (values above 1).

As a result of the club's formation, CAV 0 bears the largest negative externality, as its travel time has now increased threefold compared to the NE case (Fig. \ref{fig:5deviations}b). Surprisingly, he can now reduce his  travel time significantly by joining the club (or by deviating to Route 1). This is also true for CAV number 7, yet the impact is much smaller.

However, if the club is extended with \textit{both} players 0 and 7 (Fig. \ref{fig:5deviations}d), then it reaches internal instability (see eq. \ref{eq:internalstab}). Indeed, if players 0, 1, 5, and 6 are already on Route 1, then player 7 would actually be slower by choosing Route 1 than by sticking to the shorter route, Route 0. Marked with a down red arrow on (Fig. \ref{fig:5deviations}d) and a blue arrow denoting its transition to panel e), i.e., leaving the club. Notably, the $\{0,1,5,6,7\}$ coalition is internally unstable (as player number 7 can now increase payoffs by switching back to Route 0); the 
$\{0,1,5,6\}$ coalition is marked as stable state in Table \ref{tab:1payoffs}, where no player would unilaterally benefit by switching routes.

Finally, we examine whether HDVs gain or lose by deviating to Route 1 from the final state (Fig. \ref{fig:8human_tt}), after the stable coalition $\{0,1,5,6\}$ has been formed. HDV 13 gains from shifting to Route 1, while the other HDVs do not. After HDV 13 deviates to Route 1, the coalition remains stable (see Fig. \ref{fig:5deviations}f); no CAV wants to deviate from the club. From this final state ($\{0,1,5,6,13\}$), we do not reapply \textbf{Algorithm 1} again, though, which may reveal new coalition; thus, we refrain from calling this state a strong equilibrium in a formal sense.

Notably, the system-wide average travel times of CAVs are higher relative to the initial NE when each of the coalitions reported in Table \ref{tab:1payoffs} is formed. When the $\{1,5,6\}$ coalition is formed, the average travel times are increased by approximately 30\% (Fig. \ref{fig:6average_travel_times}a), with a slightly larger increase 
for HDVs than for CAVs. This is because all HDVs stay on Route 0, where travel times increase 
(Fig. \ref{fig:6average_travel_times}b), and the CAVs that switch to Route 1 enjoy faster travel times, indicating that the system is now dis-equilibrated (travel times on used routes are not equal).

The complexity arising from this trivial abstract example is depicted in Fig. \ref{fig:7scatter}, where each of the $1024$ possible joint actions of the game is a single dot. For illustration, we map them in an xy space according to the resulting travel times on their respective routes. Panel a) reveals a variety of complex arising states, while panel b) shows states where the coalitions were found (black circles) and which were not Nash (red circles).

\begin{figure}[H]
    \centering
    \includegraphics[width=\linewidth]{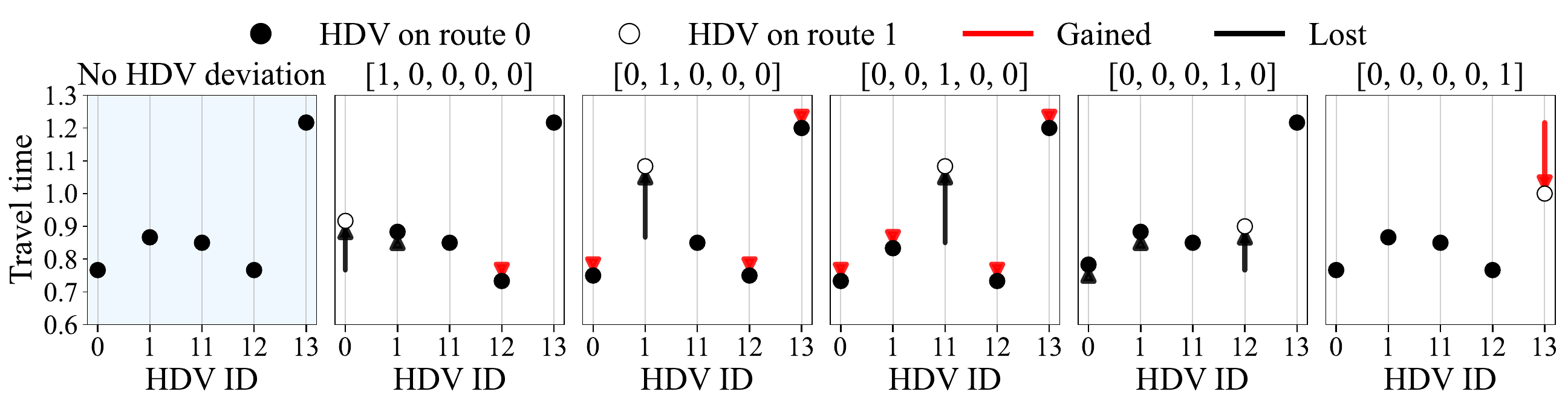}
   \caption{Changes in HDV travel times when they deviate to Route 1 after the coalition $\{0,1,5,6\}$ is formed. Only HDV 13 gains by switching to Route 1. The rest of the HDVs have increased travel times when they deviate alone to Route 1.}
   \label{fig:8human_tt}
\end{figure}

\section{Conclusions}

In this proof-of-existence study, we showed that coalition formation among CAVs can emerge in mixed-autonomy traffic networks with adaptive traffic lights. Using microscopic traffic simulation and coalitional game theory, we identified 244 parameter combinations of CAV penetration rates and traffic-light configurations under which CAV clubs can form and lead the system away from the UE.

In the studied setting, we observed that:
\begin{enumerate}
\item Equilibria may not be strong with CAVs, as coalitions can emerge.
\item Each member invited to join the club individually benefits 
by joining.
\item Not all players can be invited — the club is exclusive.
\item Clubs are beneficial for members only; other players travel 
longer, and the system becomes less efficient.
\item A FIFO-violating supply model (adaptive traffic lights) with non-monotonic supply model was 
needed to observe the phenomenon.
\item The system stabilized in a new stable state (in a Nash sense) 
until a further coalition process is triggered.
\end{enumerate}

These findings suggest that, in at least some traffic scenarios, the 
coordination capabilities of CAVs may lead to coalition formation. Both HDVs and CAVs outside the coalition may face longer travel times, and the system shifts away 
from the social optimum. Whether and how often this occurs in 
general settings remains an open question.

Such clubs may be exclusive, with only some players invited. This 
raises questions about equitable access to public road networks, 
particularly if membership in such coalitions becomes paid.

This study is deliberately minimalistic: we isolated the core properties 
of the game, neglected a series of real-world phenomena, and 
provided a basic game formulation sufficient for the existence proof. 
The specific context is also distant from classical coalitional game 
theory, as utilities (travel times) are non-transferable, so direct analogies with the core, balancedness, or Shapley values need to be applied cautiously.

Preventing the formation of such clubs would require 
privacy-preserving monitoring of collective routing patterns to detect inequitable travel-time distributions.

As a proof-of-existence study, it opens 
series of research questions. Including more 
rigorous formalization of emerging equilibrium states, broader experiments on more complex 
network topologies, scalable coalition-search algorithms, robustness 
to non-recurrent states, including competition between two or more CAV fleets, 
and policy interventions to detect and prevent exclusive club 
formation.

\backmatter


\bibliography{transportation_springer/sn-bibliography}

\end{document}